\documentclass[letterpaper,10pt,conference]{ieeeconf}

\IEEEoverridecommandlockouts
\newif\ifshowcomments
\showcommentsfalse % TOGGLE COMMENTS HERE

\usepackage{amsmath, amsfonts, amssymb, mathtools, dsfont, circledsteps, bm}

\usepackage{amsthm}
\usepackage[per-mode=symbol]{siunitx}[=v2]
\usepackage{bbm}
\usepackage{stmaryrd}
\usepackage[acronym]{glossaries-extra}
\setabbreviationstyle[acronym]{long-short}
\glssetcategoryattribute{acronym}{nohyperfirst}{true}

\usepackage{multirow}
\usepackage{booktabs}
\usepackage{array}
\newcommand{\PreserveBackslash}[1]{\let\temp=\\#1\let\\=\temp}
\newcolumntype{C}[1]{>{\PreserveBackslash\centering}p{#1}}
\newcolumntype{R}[1]{>{\PreserveBackslash\raggedleft}p{#1}}
\newcolumntype{L}[1]{>{\PreserveBackslash\raggedright}p{#1}}
\usepackage{float}
\usepackage{comment}
\usepackage{xcolor}
\usepackage{todonotes}
\usepackage{comment}
\let\labelindent\relax
\usepackage{enumitem}
\setlist[description]{font=\normalfont\itshape,noitemsep,topsep=2pt}
\usepackage[font=footnotesize]{caption}
\usepackage{subcaption}
\usepackage{cite}
\usepackage{hyperref}
\usepackage[capitalise]{cleveref}

\Crefname{equation}{Equation}{Equations}
\Crefname{definition}{Definition}{Definitions}
\Crefname{lemma}{Lemma}{Lemmas}
\Crefname{theorem}{Theorem}{Theorems}
\Crefname{proposition}{Proposition}{Propositions}
\Crefname{corollary}{Corollary}{Corollaries}
\Crefname{example}{Example}{Examples}
\creflabelformat{equation}{#2#1#3}
\Crefname{figure}{Fig.}{Figs.}

\newacronym{uav}{UAV}{unmanned aerial vehicle}

\newcommand{\cat}[1]{\mathsf{#1}}
\newcommand{\bool}{\cat{Bool}}
\newcommand{\cost}{\cat{Cost}}
\newcommand{\fuzz}{\cat{Fuzz}}
\newcommand{\pace}{\cat{Pace}}

\newcommand{\quant}{\cat{Q}}
\newcommand{\quantNats}{\bar{\cat{N}}}

\DeclareMathOperator{\prov}{prov}
\DeclareMathOperator{\req}{req}
\DeclareMathOperator{\obj}{obj}

\newcommand{\R}{\cat{R}}
\newcommand{\F}{\cat{F}}
\newcommand{\M}{\cat{M}}
\newcommand{\I}{\cat{I}}
\newcommand{\J}{\cat{J}}
\renewcommand{\phi}{\varphi}

\DeclareMathOperator{\Tr}{Tr}

\newcommand{\Reals}{\mathbb{R}}
\newcommand{\Nats}{\mathbb{N}}
\newcommand{\NatsInf}{\bar{\Nats}}
\newcommand{\maketup}[1]{\left( #1 \right)}
\newcommand{\makeset}[1]{\left\{ #1 \right\}}
\newcommand{\homarg}[4]{\llbracket {#1}, {#2} \rrbracket_{#3}}
\renewcommand{\P}{\mathcal{P}}

\newcommand{\singleton}{\star}

\makeatletter
\def\slashedarrowfill@#1#2#3#4#5{%
  $\m@th\thickmuskip0mu\medmuskip\thickmuskip\thinmuskip\thickmuskip
  \relax#5#1\mkern-7mu%
  \cleaders\hbox{$#5\mkern-2mu#2\mkern-2mu$}\hfill
  \mathclap{#3}\mathclap{#2}%
  \cleaders\hbox{$#5\mkern-2mu#2\mkern-2mu$}\hfill
  \mkern-7mu#4$%
}
\def\rightslashedarrowfill@{%
  \slashedarrowfill@\relbar\relbar\mapstochar\rightarrow}
\newcommand\xslashedrightarrow[2][]{%
  \ext@arrow 0055{\rightslashedarrowfill@}{#1}{#2}}
\makeatother
\def\slashedrightarrow{\xslashedrightarrow{}}
\def\slashedrightarrow{\relbar\joinrel\mapstochar\joinrel\rightarrow}

\definecolor{buzzgold}{HTML}{EAAA00}
\definecolor{techgold}{HTML}{B3A369}
\definecolor{navyblue}{HTML}{003057}
\definecolor{techgreen}{HTML}{22EE66}
\definecolor{mitred}{HTML}{750014}
\definecolor{caltechorange}{HTML}{FF6C0C}
\definecolor{brightred}{HTML}{FF1423}
\definecolor{afrlblue}{HTML}{003087}

\newtheorem{theorem}{Theorem}
\newtheorem{lemma}{Lemma}

\theoremstyle{definition}
\newtheorem{definition}{Definition}

\crefname{assumption}{assumption}{assumptions}
\Crefname{assumption}{Assumption}{Assumptions}

\newtheorem{example}{Example}

\theoremstyle{remark}

\renewcommand{\baselinestretch}{0.958}
\newcounter{taskcounter}

\title{\LARGE \bf 
Quantale-Enriched Co-Design: \\ Toward a Framework for Quantitative Heterogeneous System Design}

\author{Hans Riess$^{1}$, Yujun Huang$^{2}$, Matthew Klawonn$^{3}$, Gioele Zardini$^{2}$, and Matthew Hale$^{1}$% <-this % stops a space
\thanks{$^{1}$Dept. of Electrical \& Computer Engineering, Georgia Institute of Technology, Atlanta, GA, USA. Emails: \text{\{riess,mhale30\}@gatech.edu}}%
\thanks{$^{2}$Laboratory for Information \& Decision Systems (LIDS), Massachusetts Institute of Technology, Cambridge, MA, USA. Emails: \text{\{yujun233,gzardini\}@mit.edu}}%
\thanks{$^{3}$AFRL, Rome, NY, USA. Email: \text{matthew.klawonn.2@us.af.mil}}%
\thanks{Riess was supported by DARPA (HR0011-25-3-0235); Riess and Hale were supported by AFOSR (FA9550-23-1-0120, FA9550-26-1-B039) and ONR (N00014-22-1-2435, N00014-26-1-2068).
Huang and Zardini were supported by AFOSR (FA9550261B038) and DARPA YFA (D25AC00373).}
}
\begin{document}
%=-=-=-=-=-=-=-=-=-=-=-=-=-=-=-=-=-=-=-=-=-=-=-=-=-=-=-=-=-=-=-=-=-=-=-=-=-=-=%

\maketitle

\begin{abstract}
Monotone co-design enables compositional engineering design by modeling components as feasibility relations between required resources and provided functionalities. However, its boolean formulation cannot natively represent quantitative criteria such as cost, confidence, or implementation choice. In practice, these considerations are introduced through ad-hoc scalarization or by augmenting the resource space, which obscures system structure and increases computational complexity. Therefore, we develop a quantale-enriched theory of co-design in which resources and functionalities are quantale-enriched categories and design problems are quantale-enriched profunctors, which lifts co-design from boolean feasibility to quantitative evaluation. Series, parallel, and feedback composition remain valid over any commutative quantale, and change-of-base maps between quantales yield heterogeneous composition, so that subsystems evaluated in different local semantics can be composed in a common framework. Numerical examples on a target-tracking system and a UAV delivery problem demonstrate this framework and show how it avoids the drawbacks of boolean-only formulations.
\end{abstract}

%%%%%%%%%%%%%%%%%%%%%%%%%%%%%
\section{Introduction}
%%%%%%%%%%%%%%%%%%%%%%%%%%%%%
Modern cyber-physical systems, such as autonomous vehicles, robotic platforms, and mobility infrastructures, are too complex to be designed monolithically.
Their development is distributed across heterogeneous subsystems such as sensing, actuation, control, planning, power, and communication, and one would like to design, analyze, and replace components in these subsystems while retaining system-level guarantees of feasibility and other properties.
Doing so requires a formal language that is expressive enough to model heterogeneous components and compositional enough to support their interconnection.

Monotone co-design provides one such language~\cite{censiMathematicalTheoryCoDesign2016,zardiniCoDesignComplexSystems2023,censi2026CategoriesCompositionality}.
In this framework, each component is modeled by the relation between the \emph{functionalities} it provides and the \emph{resources} it requires, while system architectures are built via series, parallel, and feedback composition of such components.
The key abstraction is order-theoretic: resources and functionalities are modeled by ordered sets measuring the relative availability, quantity, demand, or preference over design choices.
A design query then asks whether a given resource can realize a given functionality, modeled by a monotone function mapping a resource-function pair $(r,f)$ to a boolean ($\mathtt{true}$ or $\mathtt{false}$). Here, monotonicity models the trade-offs accrued when choosing how to implement a design.
This perspective has enabled compositional design methods in robotics and controls, transportation, deep learning architectures, and formula-one racing~\cite{zardiniecc21,zardiniTaskdrivenModularCodesign2022,milojevic2025codei,zardini2022co,abbottFlashAttentionNapkinDiagrammatic2025,neumann2024co}.

However, in practice, feasibility is only one aspect of design evaluation. 
A component may admit several implementations, incur monetary costs, consume time or energy, or satisfy requirements only with limited confidence. 
In current practice, such quantities are often handled either by forcing them into the resource space as additional wires through ``uncurrying'' (see \Cref{fig:combined-series-design}) or by scalarizing several criteria into a single value. 
Both strategies are limiting: uncurrying obscures the architecture of the design problem and can enlarge the set of intermediate solutions that must be propagated compositionally, while scalarization hides intrinsic trade-offs and imposes arbitrary conversions between fundamentally different criteria~\cite{zardiniCoDesignComplexSystems2023}.
{
\newcommand{\figscale}{0.15}
\begin{figure}[t] % 't' or 'hp' is usually better for multi-part figures
    \centering

    % First Subfigure
    \begin{subfigure}[b]{\columnwidth}
        \centering
        \includegraphics[scale=\figscale]{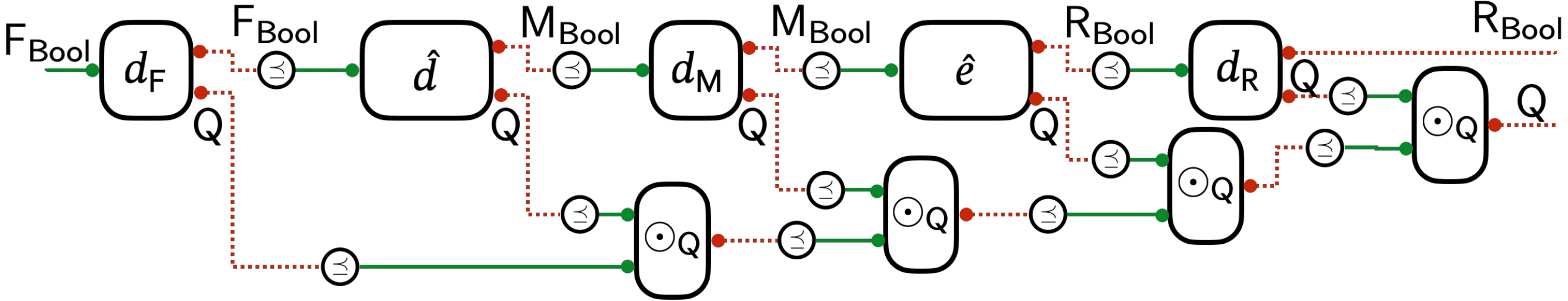}
        \caption{Uncurrying $d: \R \times \M \to \quant$ into $\tilde{d}: (\R \times \quant) \times \M \to \bool$.}
        \label{fig:hacked-co-design-series}
    \end{subfigure}    
    \hfill % Adds horizontal spacing between the two
    % Second Subfigure
    \vspace{0.5em}
    \begin{subfigure}[b]{\columnwidth}
        \centering
        \includegraphics[scale=\figscale]{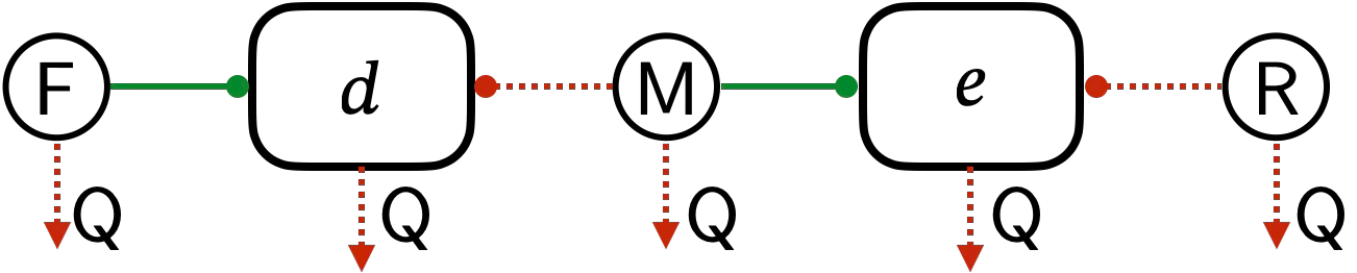}
        \caption{Quantale-enriched codesign: $d: \R \times \M \to \quant$.}
        \label{fig:q-design-series}
    \end{subfigure}
    \caption{Monotone co-design does not natively support non-boolean evaluations; uncurrying (a) obscures the underlying design problem compared with quantale-enriched modeling (b).}
    \label{fig:combined-series-design}
    \vspace{-2em}
\end{figure}
}

Therefore, this paper generalizes monotone co-design from boolean feasibility to quantale-valued evaluation.
Quantales are ordered sets equipped with a compatible binary product (\Cref{def:quantale}).
The booleans~($\bool$) are a fundamental example, but many others are relevant for design, including quantales modeling cost~($\cost$), confidence~($\fuzz$), and tuples of criteria via product quantales~\cite{rosenthalQuantalesTheirApplications1990}.
Replacing~$\bool$ by an arbitrary quantale~$\quant$ allows a component to return not only ``feasible''/``infeasible'' answers, but richer evaluations such as cost, lifespan, confidence, or combinations thereof.

In our framework, resources and functionalities are modeled as~$\quant$-enriched categories and design problems as~$\quant$-enriched profunctors.
This viewpoint subsumes several existing co-design formalisms. Classical monotone co-design is recovered when~$\quant=\bool$ and cost-aware co-design when~$\quant=\cost$ (\Cref{subsec:special}). Implementation-aware co-design can be represented using product quantales (\Cref{subsec:co-design-with-imps}).
While the observation that monotone co-design can be viewed as a profunctor between~$\bool$-enriched categories is not new~\cite{fongInvitationAppliedCategory2019}, our contribution is to develop this perspective into a practical design framework with concrete examples of relevant enrichments, heterogeneous composition rules via change of base, and accompanying computational results. 
Category theory, already present in parts of the control and systems literature~\cite{amesCategoricalLyapunovTheory2025,haghverdiBisimulationRelationsDynamical2003,speranzonAbstractionCompositionContracts2018a,sheCharacterizingCompositionalityLQR2023a,ghristCategoricalDiffusionWeighted2026}, provides the mathematical structure underlying this generalization.

A second contribution is heterogeneity.
Real systems are evaluated locally using different semantics: one component may be constrained by feasibility, another by cost, and another by implementation choice.
We introduce change-of-base maps between quantales and show that they induce compatible heterogeneous composition rules, which integrate boolean feasibility, additive cost, uncertainty, and more within a single compositional framework.

The remainder of the paper is organized as follows.
In \Cref{sec:homo}, we formalize quantale-enriched co-design and show that series, parallel, and feedback composition remain valid over arbitrary commutative quantales (\Cref{thm:traced-monoidal}). 
In \Cref{sec:hetero}, we develop heterogeneous composition through change of base (\Cref{lem:change-of-base-cat}), showing how subsystems with different local evaluation spaces can be interconnected (\Cref{thm:hetero-composition}). 
In \Cref{sec:numerical-example}, we illustrate the framework on a cost-based tracking example and a UAV case study.
Then, in Section~\ref{sec:discussion} we discuss limitations and directions for future research. 
% Proofs of all results can be found in the authors' technical report\footnote{See \url{https://hansriess.com/paper/q-design/}}.

%%%%%%%%%%%%%%%%%%%%%%%%%%%%%%%%%%%%%%
\section{Quantale-Enriched Co-Design}
\label{sec:homo}
%%%%%%%%%%%%%%%%%%%%%%%%%%%%%%%%%%%%%%
We first recall quantales and their enriched categories, and then we define quantale-enriched design problems. 

\subsection{Quantales}

A \emph{partial order}~$(Q, \sqsubseteq)$ is a set~$Q$ with a reflexive, antisymmetric, and transitive relation~$\sqsubseteq$.
Given~$S \subseteq Q$, the \emph{join} (or supremum) of~$S$, denoted~$\bigsqcup S$, is its least upper bound, and we write~$p \sqcup q$ for the join of two elements.
A partial order is a \emph{suplattice} if every subset~$S \subseteq Q$ has a join, which implies that~$Q$ has a greatest element~$\top_Q = \bigsqcup Q$ and a least element~$\bot_Q = \bigsqcup \emptyset$.

\begin{definition} \label{def:quantale}
A \emph{quantale}~$\quant = (Q, \sqsubseteq, \odot, e)$ is a suplattice~$(Q, \sqsubseteq)$ equipped with an associative, commutative binary operation $\odot$ called the \emph{multiplication}, and a unit element $e \in Q$ such that
    \begin{equation*}
        p \odot (\textstyle\bigsqcup_{i \in I} q_i) = \textstyle\bigsqcup_{i \in I} (p \odot q_i)
    \end{equation*}
    for all $p \in \quant$ and $\{q_i\}_{i \in I} \subseteq Q$.
\end{definition}

Commutativity means that $p \odot q = q \odot p$ for all $p,q \in \quant$, and the unit satisfies $e \odot p = p$ for all $p \in \quant$; commutativity is not required in general~\cite{rosenthalQuantalesTheirApplications1990}, but is assumed throughout this paper. 

\begin{example}[Boolean Feasibility] \label{eg:bool}
    The \emph{boolean quantale} $\bool = (\{ \mathtt{false}, \mathtt{true} \},\sqsubseteq, \wedge, \mathtt{true})$ models binary feasibility.
    Its order is~$\mathtt{false} \sqsubseteq \mathtt{true}$, its join is logical disjunction ($\vee$), and its multiplication is logical conjunction ($\wedge$), representing the simultaneous satisfaction of requirements.
\end{example}

\begin{example}[Cost] \label{eg:cost}
    The \emph{cost quantale} $\cost = ([0, \infty], \geq, +, 0)$ models the accumulation of quantitative resources, such as money, time, or energy.
    Note that the order is reversed:~$p \sqsubseteq q$ if $p \geq q$, making $0$ the maximal element; the join $\bigsqcup$ is the infimum ($\inf$), and multiplication 
    in the sense of Definition~\ref{def:quantale} is ordinary scalar addition ($+$).
\end{example}

\begin{example}[Confidence] \label{eg:fuzz}
    The quantale $\fuzz = ([0, 1], \leq, \odot, 1)$ models degrees of confidence, with $1$ being supported by all evidence and $0$ by no evidence; the join is the supremum ($\sup$). The product $\odot$ is a t-norm~\cite{klementTriangularNorms2000}; representative choices are the \emph{G{\"o}del} ($p \odot q = \min\{p, q\}$), \emph{Goguen} ($p \odot q = pq$), and \emph{\L{}ukasiewicz} ($p \odot q = \max\{0, p+q-1\}$) t-norms.
\end{example}

\begin{example}[Hierarchical Feasibility] \label{eg:face}
    A \emph{PACE (primary, alternate, contingency, emergency) plan}~\cite{ryanShortNotePACE2013} ranks fallback communication channels, giving the finite quantale $\pace = (\{P, A, C, E\}, \preceq_{\pace}, \wedge, P)$ with $E \preceq_{\pace} C \preceq_{\pace} A \preceq_{\pace} P$. Where $\bool$ reports whether a design is feasible, $\pace$ reports the operational state under which it remains so, and $\wedge$ selects the worst-case channel meeting a goal.
\end{example}

\begin{example}[Implementation] \label{ex:product-quantale} \label{eg:prod}
     Given a quantale $\quant$ and an indexing set $\I$, the \emph{product quantale} $\prod_{i \in \I}\quant$ is equivalent to the set of functions $\quant^{\I} \coloneqq \{\sigma: \I \to \quant\}$ with the pointwise order, join, and multiplication. When the index set $\I$ is interpreted as a set of implementations, the product quantale can be used to track feasibility, confidence, or costs of implementations via $\bool^\I$, $\fuzz^\I$, or $\cost^\I$, respectively. Notably, $\bool^\I$ is equivalent to the powerset $\P(\I)$, which forms a quantale under the usual inclusion order, set union as the join operator, and intersection as multiplication.
\end{example}

\subsection{Quantale-Enriched Categories}

Quantale-enriched categories generalize the ordered sets of resources and functionalities used in monotone co-design.

\begin{definition}
    A \emph{$\quant$-category} $\cat{C}$ consists of a set of objects $\obj(\cat{C})$ and, for every pair of objects $x, y$, a \emph{hom-value} $\homarg{x}{y}{\cat{C}}{\quant} \in \quant$ such that (i)~$e \sqsubseteq \homarg{x}{x}{\cat{C}}{\quant}$ for all $x \in \obj(\cat{C})$, and (ii)~$\homarg{x}{y}{\cat{C}}{\quant} \odot \homarg{y}{z}{\cat{C}}{\quant} \sqsubseteq \homarg{x}{z}{\cat{C}}{\quant}$ for all $x, y, z \in \obj(\cat{C})$.
    For brevity, we write $x \in \cat{C}$ for $x \in \obj(\cat{C})$.
\end{definition}

\begin{example}[Orders]
    A $\bool$-category $\cat{C}$ is an order: we write $x \preceq_{\cat{C}} y$ if and only if $\homarg{x}{y}{\cat{C}}{\bool} = \mathtt{true}$. In monotone co-design, $\bool$-categories model spaces of resources and functionalities, with $x \preceq_{\cat{C}} y$ indicating that resource $y$ is ``at least as good'' as resource $x$, whatever ``good'' means in a specific design context.
\end{example}

\begin{example}[Metrics]
    A $\cost$-category $\cat{C}$ is a generalized metric space~\cite{lawvereMetricSpacesGeneralized1973}: $\homarg{x}{y}{\cat{C}}{\cost}$ is the cost of upgrading from resource $x$ to resource $y$. The identity axiom states that upgrading a resource to itself is free, and the composition law is the triangle inequality. Metric spaces are recovered when the cost is symmetric, finite, and separated ($x \neq y \implies \homarg{x}{y}{\cat{C}}{\cost} > 0$).
\end{example}

\begin{example}[Tensor Product] \label{def:tensor-product}
     Given two $\quant$-categories $\cat{C}$ and $\cat{D}$, their \emph{tensor product} $\cat{C} \otimes \cat{D}$ is the $\quant$-category with $\obj(\cat{C} \otimes \cat{D}) \coloneqq \obj(\cat{C}) \times \obj(\cat{D})$ and pointwise hom-values $\homarg{(c, d)}{(c', d')}{\cat{C} \otimes \cat{D}}{\quant} \coloneqq \homarg{c}{c'}{\cat{C}}{\quant} \odot \homarg{d}{d'}{\cat{D}}{\quant}$.
\end{example}

\subsection{Quantale-Enriched Co-Design}

We can now replace boolean feasibility relations by quantale-valued design queries.

\begin{definition}
\label{def:dp}
Let~$\quant = (Q,\sqsubseteq,\odot,e)$ be a quantale and let~$\R$ and $\F$ be $\quant$-categories.
A \emph{$\quant$-design problem} (equivalently, a \emph{$\quant$-enriched profunctor}), written $\R \xslashedrightarrow{d}_{\quant} \F$, is a function~$d: \obj(\R) \times \obj(\F) \to \quant$ such that
    \begin{align} \label{eq:bimodule}
        \homarg{f^\ast}{f}{\F}{\quant} \odot d(r,f) \odot \homarg{r}{r^\ast}{\R}{\quant} \sqsubseteq d(r^\ast,f^\ast)
    \end{align}
    for all~$f, f^\ast \in \F$ and $r,r^\ast \in \R$.
    The objects of $\R$ are interpreted as resources, the objects of $\F$ as functionalities, and $d(r,f) \in \quant$ as a measure of how well the query $(r,f) \in \R \times \F$ is satisfied.
\end{definition}

In \eqref{eq:bimodule}, the left-hand side combines the degree $\homarg{f^\ast}{f}{\F}{\quant}$ to which the functionality can be relaxed from $f$ to $f^\ast$, the satisfaction $d(r,f)$ of the original query, and the degree $\homarg{r}{r^\ast}{\R}{\quant}$ to which the resource can be upgraded from $r$ to $r^\ast$; this composite never exceeds the satisfaction $d(r^\ast,f^\ast)$ of the improved query.

\begin{lemma} \label{lem:dp-alt}
    The condition in \eqref{eq:bimodule} holds if and only if
    \begin{align}
        \homarg{f^\ast}{f}{\F}{\quant} \odot \homarg{r}{r^\ast}{\R}{\quant} \sqsubseteq [d(r,f),d(r^\ast,f^\ast)] \label{eq:bimodule-alt}
    \end{align}
    for all $f,f^\ast \in \F$ and for all $r,r^\ast \in \R$, where $[-,-]$ is the internal hom (e.g., \cite[Sec.~3]{abramskyQuantalesObservationalLogic1993}) defined by $[p, q] \coloneqq \bigsqcup \{ r \in \quant \mid p \odot r \sqsubseteq q \}$.%
    %\footnote{For proofs of all theorems and lemmas, see our full technical report: \url{https://hansriess.com/paper/q-design/}.}
    % \footnote{For proofs of all theorems and lemmas, see the Appendix.}
\end{lemma}

The fundamental operations of monotone co-design, namely series composition, parallel composition, and feedback, generalize naturally from booleans to any quantale $\quant$.

\begin{description}
    \item[Series composition] allows us to chain design problems through an intermediate space $\M$. Given $d_1: \R \xslashedrightarrow{} \M$ and $d_2: \M \xslashedrightarrow{} \F$, the composite problem $(d_1;d_2): \R \xslashedrightarrow{} \F$ is defined as
          \begin{equation} \label{eq:q-series}
              (d_1;d_2)(r, f) \coloneqq \bigsqcup_{m \in \M} d_1(r, m) \odot d_2(m, f).
          \end{equation}
    \item[Parallel composition] describes systems operating independently. Given $d: \R \xslashedrightarrow{} \F$ and $d': \R' \xslashedrightarrow{} \F'$, their product $(d \otimes d'): \R \otimes \R' \xslashedrightarrow{} \F \otimes \F'$ is given by
          \begin{equation} \label{eq:q-parallel}
              (d \otimes d')((r, r'), (f, f')) \coloneqq d(r, f) \odot d'(r', f').
          \end{equation}
          
    \item[Feedback] (or trace) is used to model internal loops where a functionality is fed back as a resource. For a design problem $d: \R \otimes \M \xslashedrightarrow{} \F \otimes \M$, the feedback $\Tr_{\M}(d): \R \xslashedrightarrow{} \F$ is defined as $\Tr_{\M}(d)(r, f) \coloneqq$
          \begin{equation} \label{eq:q-feedback}
              \begin{aligned}
                   \textstyle\bigsqcup_{m, m' \in \M} d((r, m), (f, m')) \odot \homarg{m}{m'}{\M}{\quant}.
              \end{aligned}
          \end{equation}
          The trace jointly evaluates the component $d\bigl((r,m),(f,m')\bigr)$ and the degree $\homarg{m}{m'}{\M}{\quant}$ to which the internal output $m'$ satisfies the internal input $m$.
\end{description}

\begin{theorem} \label{thm:traced-monoidal}
    For any quantale $\quant$, the composition operations in \eqref{eq:q-series}, \eqref{eq:q-parallel}, 
    and~\eqref{eq:q-feedback} all produce valid $\quant$-design problems.
\end{theorem}

In fact, $\quant$-categories (as objects) and $\quant$-design problems (as morphisms) form a compact closed, hence traced monoidal~\cite{joyalTracedMonoidalCategories1996}, category~\cite[Ch.~4]{fongInvitationAppliedCategory2019}. \Cref{thm:traced-monoidal} verifies directly that the three operations used in this paper are well-defined.
Consequently, interconnected $\quant$-design problems can be manipulated using string diagrams~\cite{fongInvitationAppliedCategory2019}, as in monotone co-design~\cite{zardiniCoDesignComplexSystems2023}.

\subsection{\texorpdfstring{Monotone and Quantitative Co-Design}{Monotone and Quantitative Co-Design}} \label{subsec:special}

Monotone co-design~\cite{zardiniCoDesignComplexSystems2023} is precisely recovered by $\bool$-design problems: $\R$ and $\F$ are ordered sets of resources and functionalities, $d(r,f) = \mathtt{true}$ if functionality $f$ is feasible given resource $r$, and \eqref{eq:bimodule} states that feasibility is preserved when moving to a lesser functionality or a greater resource.
Series composition in \eqref{eq:q-series} becomes $\bigvee_{m \in \M} d_1(r,m) \wedge d_2(m,f)$, which holds if and only if some intermediate $m$ is both feasible given $r$ and sufficient for $f$. Parallel composition is the conjunction of the two feasibilities. The feedback in \eqref{eq:q-feedback} becomes $\bigvee_{m \preceq_{\M} m'} d\bigl((r,m),(f,m')\bigr)$, since $\homarg{m}{m'}{\M}{\bool}$ is the order relation $m \preceq_{\M} m'$.
Because $d$ is monotone increasing in required resources, this supremum reduces to $\bigvee_{m \in \M} d\bigl((r,m),(f,m)\bigr)$, the existence of an internal resource $m$ that sustains the loop.

Similarly, $\cost$-design problems evaluate designs subject to costs.
The hom-object $\homarg{r}{r'}{\R}{\cost} \in [0,\infty]$ is the minimum cost of converting resource $r$ into $r'$, so the identity axiom makes leaving a resource unchanged free, 
i.e., $\homarg{r}{r}{\R}{\cost} = 0$. The composition axiom is the triangle inequality $\homarg{r}{r'}{\R}{\cost} + \homarg{r'}{r''}{\R}{\cost} \ge \homarg{r}{r''}{\R}{\cost}$. The $\cost$-category $\F$ of functionalities is read analogously.
A $\cost$-design problem assigns to each pair the cost $d(r,f) \in [0,\infty]$ of a component providing $f$ from $r$, and \eqref{eq:bimodule} becomes the bound
\begin{equation*}
    \homarg{f^\ast}{f}{\F}{\cost} + d(r,f) + \homarg{r}{r^\ast}{\R}{\cost} \ge d(r^\ast,f^\ast),
\end{equation*}
while series composition is a minimization over the interface: 
\begin{equation*}
    (d_1;d_2)(r,f) = \inf_{m \in \M} d_1(r,m) + d_2(m,f).
\end{equation*}

%%%%%%%%%%%%%%%%%%%%%%%%%%%%%%%%%%%
\section{Heterogeneous Co-Design}
\label{sec:hetero}
%%%%%%%%%%%%%%%%%%%%%%%%%%%%%%%%%%%

Designers often evaluate different parts of a design using different quantales.
This section relaxes the requirement of \Cref{sec:homo} that interconnected models share a single evaluation space.

%%%%%%%%%%%%%%%%%%%%%%%%%%%%%%%%%%%
\subsection{Change of Base}
\label{sec:change-of-base}
%%%%%%%%%%%%%%%%%%%%%%%%%%%%%%%%%%%

We formalize mappings between different evaluation spaces using structure-preserving maps between quantales, a procedure known as ``change of base'' \cite{borceuxHandbookCategoricalAlgebra1994}.

\begin{definition} \label{def:lax}
    Suppose $\quant = (Q,\sqsubseteq,\odot,e)$ and $\quant' = (Q',\sqsubseteq',\odot',e')$ are quantales.
    A \emph{lax function} $\phi: \quant \to \quant'$ between quantales is a monotone function $\phi: (Q,\sqsubseteq) \to (Q',\sqsubseteq')$ such that
    \begin{align} \label{eq:lax}
        \phi(q_1) \odot' \phi(q_2) \sqsubseteq' \phi(q_1 \odot q_2) \quad\text{and}\quad e' \sqsubseteq' \phi(e)
    \end{align}
    for all $q_1, q_2 \in Q$.
    The function $\phi$ is said to be \emph{strict} if \eqref{eq:lax} holds with equality and \emph{oplax} if the orders are reversed.
\end{definition}

\begin{example}[Cost-to-Boolean]\label{eg:quant2bool}
    Projecting cost-based evaluations into boolean feasibility calls for a lax function $\phi: \cost \to \bool$. By \Cref{def:lax}, the feasible set $\{ c \in [0, \infty] \mid \phi(c) = \mathtt{true} \}$ contains $0$, is downward closed in the standard real order, and is closed under addition, leaving only $[0,\infty)$, $\{0\}$, and $[0,\infty]$: any finite cost is feasible, only free implementations are feasible, or everything is feasible. Hence no lax function encodes a finite positive budget, and structure-preserving maps $\cost \to \bool$ admit only absolute, and potentially conservative, thresholds.
\end{example}

\begin{example}[Cost-to-Cost]\label{eg:cost2cost}
    Any monotone concave function $\phi: [0,\infty) \to [0,\infty)$ with $\phi(0) = 0$ extends  to a lax function on $\cost$ by setting $\phi(\infty) = \infty$.
    This extension is unique if $\phi$ is unbounded. For example, $\phi_K(r) = K \cdot r$ for any $K > 0$ is strict. Such maps change units, e.g., from dollars to Swiss francs.
\end{example}

\begin{example}[From Boolean] \label{eg:bool2quant}
    Conversely, boolean feasibility constraints must often be embedded into a richer quantale $\quant$ so that they can be composed with quantitative components. The canonical change of base is the strict function $\phi: \bool \to \quant$ mapping $\mathtt{true}$ to the unit $e_{\quant}$ and $\mathtt{false}$ to the bottom element $\bot_{\quant}$. For $\quant = \cost$, where $e_{\cost} = 0$ and $\bot_{\cost} = \infty$, this change of base yields
    \begin{equation} \label{eq:bool2cost}
        \phi(\mathtt{true}) = 0, \qquad \phi(\mathtt{false}) = \infty.
    \end{equation}
    
\end{example}

\begin{lemma}[{\cite[Chapter 6.4]{borceuxHandbookCategoricalAlgebra1994}}] \label{lem:change-of-base-cat}
    Suppose $\quant$ and $\quant'$ are quantales and $\phi: \quant \to \quant'$ is a monotone function. Define the \emph{pushforward} data $\phi_{\ast}\cat{C}$ by $\obj(\phi_{\ast}\cat{C}) \coloneqq \obj(\cat{C})$ and $\homarg{x}{y}{\phi_{\ast}\cat{C}}{\quant'}\coloneqq \phi\bigl(\homarg{x}{y}{\cat{C}}{\quant}\bigr)$. Then $\phi$ is lax if and only if, for every $\quant$-category $\cat{C}$, the pushforward $\phi_{\ast}\cat{C}$ is a $\quant'$-category.
\end{lemma}

%%%%%%%%%%%%%%%%%%%%%%%%%%%%%%%%%%%%%%%%%%%%%%%%%
\subsection{Heterogeneous Composition}
\label{sec:hetero-comp}
%%%%%%%%%%%%%%%%%%%%%%%%%%%%%%%%%%%%%%%%%%%%%%%%%

A \emph{heterogeneous design problem} is an interconnected network of components, each modeled as a $\quant_i$-design problem $d_i: \R_i \xslashedrightarrow{}_{\quant_i} \F_i$, where the evaluation quantales $\quant_i$ may differ between components. To compose them, the quantales must be related via lax functions.

\begin{lemma}\label{lem:heterogenous-co-design}
    Given a $\quant_f$-category $\F$, a $\quant_r$-category $\R$,  and lax-functions $\phi_f \colon \quant_f \to \quant$ and $\phi_r \colon \quant_r \to \quant$, a map $d \colon \obj(\R) \times \obj(\F) \to \quant$ satisfying
    \begin{equation}\label{eq:heterogenous-co-design}
        \phi_f(\homarg{f^\ast}{f}{\F}{\quant_f}) \odot d(r,f) \odot \phi_r(\homarg{r}{r^\ast}{\R}{\quant_r}) \sqsubseteq d(r^\ast,f^\ast)
    \end{equation}
    defines a $\quant$-design problem $(\phi_r)_\ast\R \xslashedrightarrow{d} (\phi_f)_\ast\F$.
\end{lemma}

For example, by using the lax function in \eqref{eq:bool2cost}, design problems defined natively over orders can be elevated into quantitative spaces.

Heterogeneous design problems can thus be composed. Suppose $d: \R \xslashedrightarrow{}_{\quant_d} \M$ and $e: \M \xslashedrightarrow{}_{\quant_e} \F$ take values in $\quant_d$ and $\quant_e$, respectively. To couple them, we fix a common evaluation quantale $\quant_{\mathsf{eval}}$ along with lax functions $\phi_d: \quant_d \to \quant_{\mathsf{eval}}$ and $\phi_e: \quant_e \to \quant_{\mathsf{eval}}$. Then:
\begin{description}
    \item[Heterogeneous series composition] yields $(d; e): (\phi_d)_\ast\R \xslashedrightarrow{}_{\quant_{\mathsf{eval}}} (\phi_e)_\ast\F$, defined pointwise by
        \begin{equation} \label{eq:hetero-series}
            (d; e)(r,f) \coloneqq \bigsqcup_{m \in \M} \phi_d\bigl(d(r, m)\bigr) \odot_{\mathsf{eval}} \phi_e\bigl(e(m, f)\bigr).
        \end{equation}
    \item[Heterogeneous parallel composition] of $d: \R \xslashedrightarrow{}_{\quant_d} \F$ and $e: \R' \xslashedrightarrow{}_{\quant_e} \F'$ yields $(d \otimes e): (\phi_d)_\ast\R \otimes (\phi_e)_\ast\R' \xslashedrightarrow{}_{\quant_{\mathsf{eval}}} (\phi_d)_\ast\F \otimes (\phi_e)_\ast\F'$, defined by
        \begin{equation} \label{eq:hetero-parallel}
            (d \otimes e)\left((r, r'),(f, f')\right) \coloneqq \phi_d\bigl(d(r, f)\bigr) \odot_{\mathsf{eval}} \phi_e\bigl(e(r', f')\bigr).
        \end{equation}
    \item[Heterogeneous feedback] of $d: \R \otimes \M \xslashedrightarrow{}_{\quant_d} \F \otimes \M$ yields $\Tr_{\M}(d): (\phi_d)_\ast\R \xslashedrightarrow{}_{\quant_{\mathsf{eval}}} (\phi_d)_\ast\F$, defined by
        \begin{equation} \label{eq:hetero-trace}
            \begin{aligned}
                \Tr_{\M}(d)(r, f) \coloneqq \bigsqcup_{m, m' \in \M} & \phi_d\bigl(d((r, m), (f, m'))\bigr)                                  \\
                                                                     & \odot_{\mathsf{eval}} \phi_d\bigl(\homarg{m}{m'}{\M}{\quant_d}\bigr).
            \end{aligned}
        \end{equation}
\end{description}

\begin{theorem}\label{thm:hetero-composition}
    The heterogeneous composition operations defined in \eqref{eq:hetero-series},
    \eqref{eq:hetero-parallel},
    and~\eqref{eq:hetero-trace} produce valid quantale-enriched design problems.
\end{theorem}

\subsection{Co-Design with Implementation}\label{subsec:co-design-with-imps}

To illustrate heterogeneous composition, consider components that are each realized by selecting an implementation $i$ from a catalog $\I$~\cite[Ch.~29]{censi2026CategoriesCompositionality}, which requires resources $\req_\I(i) \in \R$ and provides functionalities $\prov_\I(i) \in \M$.
Evaluating feasibility in the powerset quantale $\quant_d \coloneqq \P(\I)$ yields the profunctor $\R \xslashedrightarrow{d}_{\P(\I)} \M$ defined for all $r \in \R$ and $m \in \M$ as
\begin{equation*}
    d(r, m) = \{ i \in \I \mid \prov_\I(i) \succeq_{\M} m \wedge \req_\I(i) \preceq_{\R} r \},
\end{equation*}
where the orders $\R$ and $\M$ are regarded as $\P(\I)$-categories via \Cref{eg:bool2quant}.
A second component with implementations $\J$ similarly yields $\M \xslashedrightarrow{e}_{\P(\J)} \F$ with $e(m, f) = \{ j \in \J \mid \prov_\J(j) \succeq_{\F} f \wedge \req_\J(j) \preceq_{\M} m \}$.

Because $d(r,m) \in \P(\I)$ and $e(m,f) \in \P(\J)$ reside in different evaluation spaces, they cannot be composed directly.
The shared quantale is the product implementation space $\quant_{\mathsf{eval}} \coloneqq \P(\I \times \J)$, into which both are projected via the lax functions $\phi_d(A) \coloneqq A \times \J$ and $\phi_e(B) \coloneqq \I \times B$ for $A \subseteq \I$ and $B \subseteq \J$.
Applying \eqref{eq:hetero-series}, we obtain the composite profunctor
\begin{align*}
    (d ; e)(r,f) & = \bigcup_{m \in \M} \bigl(d(r,m) \times \J\bigr) \cap \bigl(\I \times e(m,f)\bigr)            \\
                 & = \bigcup_{m \in \M} d(r,m) \times e(m,f).
\end{align*}
The pairs $(i,j) \in (d ; e)(r,f)$ are exactly those for which some intermediate interface $m \in \M$ satisfies $\prov_\I(i) \succeq_{\M} m \succeq_{\M} \req_\J(j)$.

%%%%%%%%%%%%%%%%%%%%%%%%%%%%%%%%%%%%%%%%%%%
\section{Numerical Examples}
\label{sec:numerical-example}
%%%%%%%%%%%%%%%%%%%%%%%%%%%%%%%%%%%%%%%%%%%%%%%%
We illustrate our framework on (i) a small~$\cost$-enriched tracking problem, verified analytically and with the \texttt{mcdp} solver~\cite{censiMathematicalTheoryCoDesign2016,zardiniCoDesignComplexSystems2023} and (ii) a heterogeneous \gls{uav} delivery case study in which one architecture is queried in~$\cost$ (minimum lifetime cost) and in~$\P(\I_a\times\I_b)$ (feasible actuator--battery implementations under a budget).

\subsection{Quantitative Co-Design Example}
We consider the design of a target-tracking system, adapted from~\cite{huang2026distributionaluncertaintyadaptivedecisionmaking}: given an onboard power budget, what is the monetary cost of a desired tracking capacity?
The system is the series composition of a sensor $d_{\text{sensor}}: \R \xslashedrightarrow{} \M$ and a processor $d_{\text{proc}}: \M \xslashedrightarrow{} \F$, where the power budgets $R = \{5\text{W}, 10\text{W}, 20\text{W}\}$, the sensor data qualities $M = \{\text{Low}, \text{High}\}$, and the tracking capacities $F = \{1, 2, 3\}$ are ordered in the obvious way, and, following \Cref{lem:heterogenous-co-design}, regarded as $\cost$-categories via \eqref{eq:bool2cost}.
The entry $d_{\text{sensor}}(r, m)$ is the cost of a sensor yielding data quality $m$ within power budget $r$, and $d_{\text{proc}}(m, f)$ is the cost of the processing stack required to track $f$ targets from data quality $m$ (rows: resources; columns: functionalities):
\begin{equation*}
    \begin{array}{c|cc}
        d_{\text{sensor}} & \text{Low} & \text{High} \\
        \hline
        5\text{W}         & \$30       & \infty      \\
        10\text{W}        & \$20       & \$50        \\
        20\text{W}        & \$10       & \$30
    \end{array}
    \qquad
    \begin{array}{c|ccc}
        d_{\text{proc}} & 1            & 2            & 3       \\
        \hline
        \text{Low}      & \$40         & \$70         & \infty  \\
        \text{High}     & \$10         & \$30         & \$50
    \end{array}
\end{equation*}
Both matrices satisfy \cref{lem:dp-alt}: improving the available resource or relaxing the requested functionality cannot increase the cost.
The composite $(d_{\text{sensor}}; d_{\text{proc}})$ is computed by \eqref{eq:q-series}, where $\bigsqcup = \inf$ and $\odot = +$ in $\cost$:
\begin{equation*}
    \begin{array}{c|ccc}
        d_{\text{sensor}}; d_{\text{proc}} & 1     & 2      & 3      \\
        \hline
        5\text{W}                          & \$70  & \$100  & \infty \\
        10\text{W}                         & \$60  & \$80   & \$100  \\
        20\text{W}                         & \$40  & \$60   & \$80
    \end{array}
\end{equation*}
Each entry is an optimization over the interface $\M$. For example, tracking 2 targets on a 10W budget costs $\min\{\$20 + \$70,\; \$50 + \$30\} = \$80$: the more expensive ``High'' quality sensor is optimal because it reduces the processing cost.
We verified these results with the \texttt{mcdp} solver~\cite{censiMathematicalTheoryCoDesign2016,zardiniCoDesignComplexSystems2023} by encoding the tracking node in $\bool$ with the cost uncurried into an explicit resource wire, as in \cref{fig:hacked-co-design-series}.
Querying for Pareto-optimal (power, cost) pairs under each target load returns the minimal upward-closed sets
\begin{align*}
    \text{1 target:}  & \quad \uparrow\!\{( 5\,\text{W},\,\$70),\, ( 10\,\text{W},\,\$60),\, ( 20\,\text{W},\,\$40)\},  \\
    \text{2 targets:} & \quad \uparrow\!\{( 5\,\text{W},\,\$100),\, ( 10\,\text{W},\,\$80),\, ( 20\,\text{W},\,\$60)\}, \\
    \text{3 targets:} & \quad \uparrow\!\{( 10\,\text{W},\,\$100),\, ( 20\,\text{W},\,\$80)\},
\end{align*}
which agree with the composite matrix above. The 5\,W budget is infeasible at 3 targets, consistent with $d_{\text{sensor}}(5\,\text{W},\text{High}) = \infty$.

\subsection{Heterogeneous \texorpdfstring{\gls{uav}}{UAV} Case Study}\label{subsec:uav-q-design-models}

We next consider a \gls{uav} delivery design problem adapted from~\cite{huang2026distributionaluncertaintyadaptivedecisionmaking}: the component models and catalogs are taken from~\cite{huang2026distributionaluncertaintyadaptivedecisionmaking,censi2017uncertainty}, while the quantale-enriched formulations and the two queries below are new.

The \gls{uav} must fulfill deliveries within a region.
The number of deliveries fulfilled is a functionality modeled with the quantale of natural numbers with infinity, $\quantNats = \maketup{\NatsInf,\geq,+,0}$, which is also a category enriched in itself, with $\homarg{n}{m}{\quantNats}{\quantNats}$ equal to $0$ when $n \geq m$ and to $m-n$ (the number of deliveries not served) otherwise.
A lax function $\phi_{\quantNats, \cost}$ then represents the \emph{cost of deliveries not served}; laxity, $\phi_{\quantNats, \cost}(k_1) + \phi_{\quantNats, \cost}(k_2) \geq \phi_{\quantNats, \cost}(k_1 + k_2)$, means a decreasing marginal cost per unserved delivery.
We require $1000$ deliveries and take $\phi_{\quantNats, \cost}(k) = M \sqrt{k}$ with $M=2$, so that forgoing all deliveries costs $2\sqrt{1000} \approx \$63$, comparable to the component costs in \cref{fig:uav-q-design-payload-cost}.
The maximum delivery distance is \SI{1000}{\meter} and the delivery frequency is one per \SI{500}{\second}; the payload capacity, a second functionality, is varied between \SI{100}{\gram} and \SI{2900}{\gram}.

The \gls{uav} decomposes into task management, perception, actuation, and battery.
Although the models are simplified, compositional co-design allows for more detailed models with little overhead \cite{zardiniTaskdrivenModularCodesign2022}.
\emph{Task management} converts the number of deliveries satisfied (possibly fewer than $1000$, at a cost), the delivery distance, and the frequency into required battery cycles, endurance, and velocity.
\emph{Perception} is a fixed sensor and software stack, with no design choice, whose power consumption increases with the required velocity~\cite{karaman2012high}.
\emph{Actuation} is a free choice among three actuators $\I_a = \makeset{a_1, a_2, a_3}$, each with a weight, cost, and maximum velocity, and each drawing power $P \geq p_0 + p_1 F^2$ from the battery to provide lift force $F$.
\emph{Battery} design selects among eight technologies, forming the implementation set $\I_b$, each with an energy density, unit power per cost, and number of cycles per maintenance, together with the battery weight.

\subsubsection{\texorpdfstring{$\cost$-enriched model}{Cost}}\label{subsec:uav-model-quantale-cost}

Suppose we seek the minimum \emph{lifetime cost} of a design, modeled with the quantale~$\cost$; the co-design diagram is shown in \cref{fig:uav-quantale-co-design-cost}.
The actuation design problem $d_A \colon \Reals_{\geq 0}^2 \times \Reals_{\geq 0}^2 \to \cost$ maps each velocity--lift functionality and power--weight resource to the minimum cost of an actuator that satisfies them. The battery design problem computes the purchase and maintenance cost of each technology that provides the desired functionality with the given resource and selects the cheapest.
The composed design problem $d_{\text{UAV},\cost} \colon \makeset{\singleton} \slashedrightarrow_\cost \Reals_{\geq 0}$ has payload as its functionality and no resources, represented by the one-object category $\makeset{\singleton}$.

\begin{figure}[tb]
    \centering
    \includegraphics[width=0.85\linewidth]{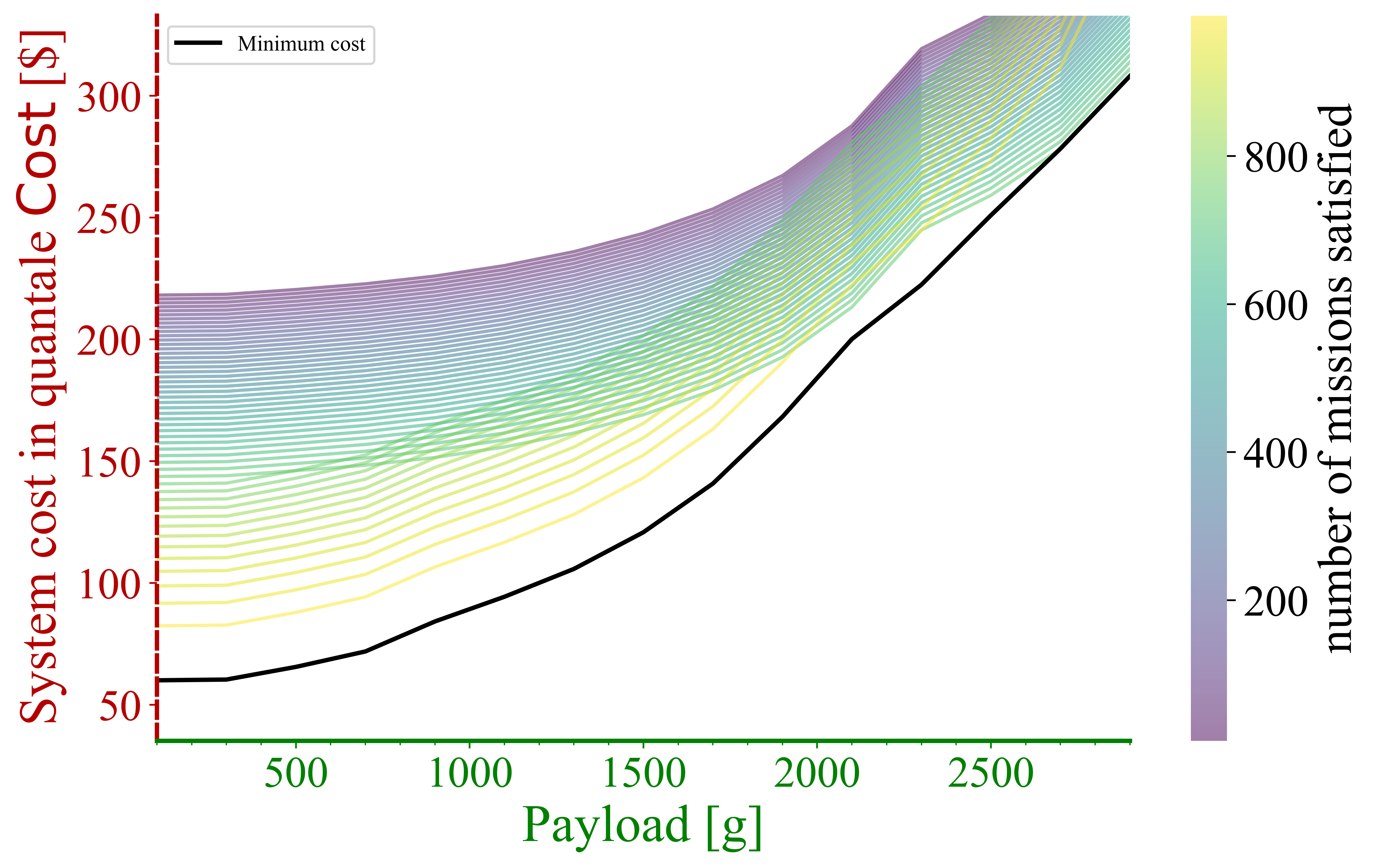}
    \caption{The cost resulting from different numbers of deliveries satisfied (colored curves) and the minimal cost as the design result in the quantale $\cost$ (black curve) for various payloads.}
    \label{fig:uav-q-design-payload-cost}
    \vspace{-1em}
\end{figure}

\begin{figure}[tb]
    \centering
    \includegraphics[width=\columnwidth]{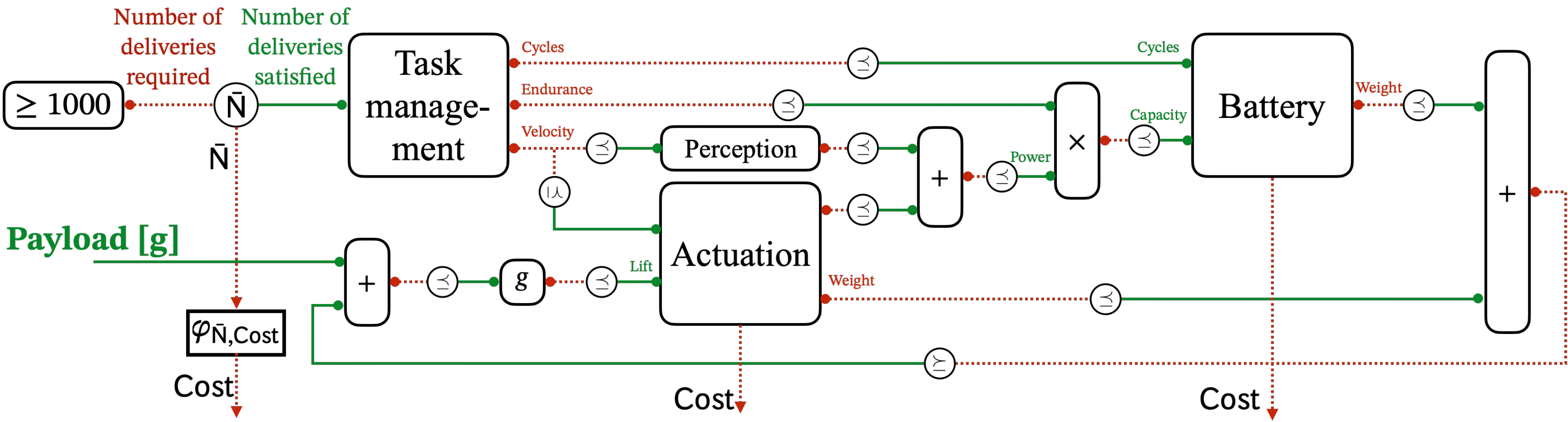}
    \caption{The $\cost$-enriched co-design diagram of the \gls{uav}. Rounded boxes are design problems; battery and actuation are enriched in $\cost$ (downward red arrows), the others in $\bool$ and composed in $\cost$ via the change of base \eqref{eq:bool2cost}. Deliveries satisfied are modeled with the $\quantNats$-category $\quantNats$, and $\phi_{\quantNats, \cost}(k) = M \sqrt{k}$ gives the induced cost.}
    \label{fig:uav-quantale-co-design-cost}
    \vspace{-1em}
\end{figure}

\Cref{fig:uav-q-design-payload-cost} shows the lifetime cost for each number of deliveries satisfied and each payload.
For each payload, the series composition \eqref{eq:q-series} selects the number of deliveries that minimizes the cost, since $\bigsqcup$ is $\inf$ and $\odot$ is summation in $\cost$.

\subsubsection{Powerset-enriched model}\label{subsec:uav-model-quantale-powerset}

\begin{figure}[h]
    \centering
    \includegraphics[width=\columnwidth]{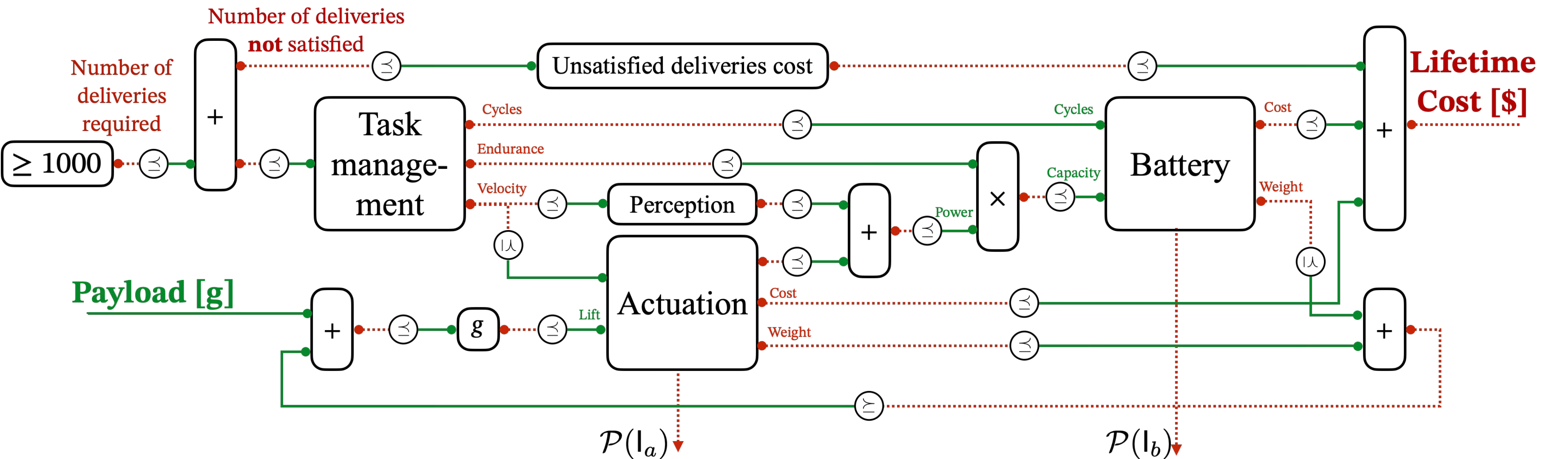}
    \caption{The \gls{uav} model enriched in the powerset of implementations. Each inequality is an order ($\bool$-category) and each design problem without implementation is a $\bool$-valued feasibility relation, mapped to $\P(\I_a \times \I_b)$ by the embedding of \Cref{eg:bool2quant} ($\mathtt{true} \mapsto \I_a \times \I_b$, $\mathtt{false} \mapsto \emptyset$).}
    \label{fig:uav-quantale-co-design-powerset-imps}
    \vspace{-1em}
\end{figure}

Suppose instead that a lifetime cost budget is given and we ask which actuators and battery technologies satisfy the payload requirement within it.
This leads to a co-design problem enriched in the powerset of implementations, illustrated in \cref{fig:uav-quantale-co-design-powerset-imps}: the number of unsatisfied deliveries is extracted explicitly as a functionality/resource, its cost is added to the battery and actuator costs to form the lifetime cost as a system-level resource, and the battery and actuation design problems evaluate to the sets of implementations satisfying each functionality--resource pair.
Fixing a payload and a lifetime cost, the designer obtains a value in $\P(\I_a \times \I_b)$ according to \cref{subsec:co-design-with-imps}. The regions where each actuator and battery technology are feasible are shown in \cref{fig:uav-q-design-payload-cost-imps-result}.

Comparing \cref{fig:uav-quantale-co-design-cost,fig:uav-quantale-co-design-powerset-imps} exposes a computational pitfall of encoding cost in $\bool$: in \cref{fig:uav-quantale-co-design-powerset-imps} the unsatisfied delivery cost becomes an explicit resource wire, which adds a wire to nearly every cut of the diagram and increases the dimension of the intermediate functionality/resource spaces, and hence the number of intermediate solutions and the solve time.
Both diagrams were solved with \texttt{mcdp}, i.e., in $\bool$ with cost uncurried into a resource wire as in \cref{fig:hacked-co-design-series}. For \cref{fig:uav-quantale-co-design-cost}, the number of deliveries is fixed in each solve and the join in \eqref{eq:q-series} over deliveries is taken afterwards.
Solving \cref{fig:uav-quantale-co-design-cost} with a fixed number of satisfied deliveries takes less than a minute, while solving \cref{fig:uav-quantale-co-design-powerset-imps} to the same resolution takes up to twenty minutes.

\begin{figure}[tb]
    \centering
    \includegraphics[width=\linewidth]{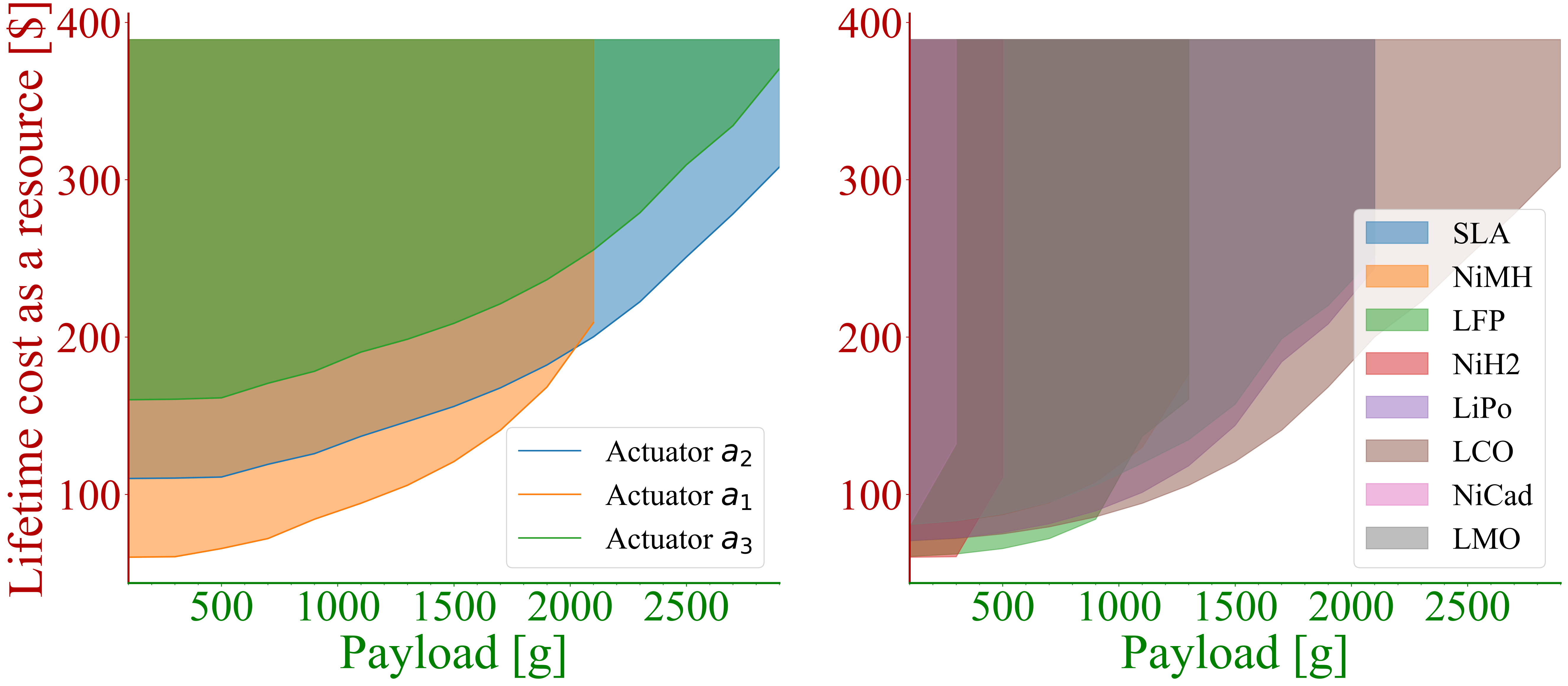}
    \caption{Regions in the payload-lifetime cost space where each actuator and battery technology are feasible. The implementation sets are plotted separately by projecting elements in $\I_a \times \I_b$ back to $\I_a$ and $\I_b$.}
    \label{fig:uav-q-design-payload-cost-imps-result}
    \vspace{-1em}
\end{figure}

%%%%%%%%%%%%%%%%%%%%%%%%%%%%%%%%%%%%%%%%%%%%%%%%
\section{Discussion}
\label{sec:discussion}
%%%%%%%%%%%%%%%%%%%%%%%%%%%%%%%%%%%%%%%%%%%%%%%%

This paper generalized monotone co-design from boolean feasibility to quantale-valued evaluation, established series, parallel, and feedback composition in this setting, and showed how heterogeneous components can be composed through changes of base.
We close with limitations and a direction that extends the framework.

\paragraph*{Scalability and Modeling Overhead}
The examples above are deliberately small.
The monotone co-design solver scales linearly in the number of components when the number of intermediate solutions is bounded~\cite{censiMathematicalTheoryCoDesign2016,zardiniCoDesignComplexSystems2023}, but that number can grow exponentially with the number of wires in a cut of the diagram. \Cref{subsec:uav-q-design-models} shows that quantale enrichment removes wires that boolean encodings must add.
A native quantale-enriched solver, and controlled experiments isolating the feedback operator and heterogeneous versus homogeneous enrichment, are left for future work.
Choosing quantales and change-of-base maps is additional modeling effort, but it is local to each component and interface, and \Cref{eg:quant2bool,eg:bool2quant} show that the admissible maps are often few and canonical.

\paragraph*{Optimal Co-Design}
A natural next step is to optimize over the feasible implementation sets of \Cref{subsec:co-design-with-imps} by changing base from~$\P(\I)$ to~$\cost$.
The most direct candidate, mapping~$S\in\mathcal{P}(\I)$ to the cost of its cheapest element, $\phi(S)=\inf_{i\in S} c(i)$, is not a lax function, so combining implementation-level composition with optimization remains an open problem that points to a deeper integration of feasibility, choice, and cost in co-design.

%%%%%%%%%%%%%%%%%%%%%%%%%%%%%%
\bibliographystyle{ieeetr}
\bibliography{biblio}
%%%%%%%%%%%%%%%%%%%%%%%%%%%%%%

%%%%%%%%%%%%%%%%%%%%%%%%%%%%%%
\newpage
\appendix
\section{Appendix}
\label{appendix}
%!TEX root = main.tex

\subsection*{Proof of \Cref{lem:dp-alt}}

We show that \eqref{eq:bimodule} and \eqref{eq:bimodule-alt} are equivalent via the defining adjunction of the internal hom. Recall that for any $p, q, r \in \quant$, the internal hom satisfies $r \odot p \sqsubseteq q$ if and only if $r \sqsubseteq [p, q]$.

Applying the adjunction to \eqref{eq:bimodule} with $p = d(r,f)$, $r = \homarg{f^\ast}{f}{\F}{\quant} \odot \homarg{r}{r^\ast}{\R}{\quant}$, and $q = d(r^\ast,f^\ast)$, we obtain
\begin{equation*}
    \homarg{f^\ast}{f}{\F}{\quant} \odot d(r,f) \odot \homarg{r}{r^\ast}{\R}{\quant} \sqsubseteq d(r^\ast,f^\ast)
\end{equation*}
if and only if
\begin{equation*}
    \homarg{f^\ast}{f}{\F}{\quant} \odot \homarg{r}{r^\ast}{\R}{\quant} \sqsubseteq [d(r,f),\, d(r^\ast,f^\ast)].
\end{equation*}

\subsection*{Proof of \Cref{thm:traced-monoidal}}
We establish the closure of 
the series, parallel, and feedback operations. 

\begin{lemma} \label{lem:q-series}
    The series composition defined in \eqref{eq:q-series} produces a valid $\quant$-design problem.
\end{lemma}
\begin{proof}
    Let $d_1: \R \xslashedrightarrow{} \M$ and $d_2: \M \xslashedrightarrow{} \F$ be two $\quant$-design problems. By \Cref{def:dp}, they satisfy the respective profunctor monotonicity conditions
    \begin{equation} \label{eq:pf_d1}
        \homarg{m^\ast}{m}{\M}{\quant} \odot d_1(r,m) \odot \homarg{r}{r^\ast}{\R}{\quant} \sqsubseteq d_1(r^\ast, m^\ast)
    \end{equation}
    and
    \begin{equation} \label{eq:pf_d2}
        \homarg{f^\ast}{f}{\F}{\quant} \odot d_2(m,f) \odot \homarg{m'}{m^\ast}{\M}{\quant} \sqsubseteq d_2(m',f^\ast)
    \end{equation}
    for any $r, r^\ast \in \R$, $m, m^\ast, m' \in \M$, and $f, f^\ast \in \F$. Setting $m = m^\ast$ in \eqref{eq:pf_d1} and using the identity axiom $e \sqsubseteq \homarg{m}{m}{\M}{\quant}$ yields
    \begin{equation} \label{eq:pf_d1_simp}
        d_1(r,m) \odot \homarg{r}{r^\ast}{\R}{\quant} \sqsubseteq d_1(r^\ast, m).
    \end{equation}
    Similarly, setting $m' = m^\ast = m$ in \eqref{eq:pf_d2} yields 
    \begin{equation} \label{eq:pf_d2_simp}
        \homarg{f^\ast}{f}{\F}{\quant} \odot d_2(m,f) \sqsubseteq d_2(m,f^\ast).
    \end{equation}

    We must show that the composite $(d_1;d_2)$ satisfies the condition given in \eqref{eq:bimodule} that 
    \begin{equation*}
        \homarg{f^\ast}{f}{\F}{\quant} \odot (d_1;d_2)(r,f) \odot \homarg{r}{r^\ast}{\R}{\quant} \sqsubseteq (d_1;d_2)(r^\ast,f^\ast).
    \end{equation*}
    Applying the definition from \eqref{eq:q-series} gives 
    \begin{align*}
         & \homarg{f^\ast}{f}{\F}{\quant} \odot \left( \bigsqcup_{m \in \M} d_1(r,m) \odot d_2(m,f) \right) \odot \homarg{r}{r^\ast}{\R}{\quant} \\
         & = \bigsqcup_{m \in \M} \Big( \homarg{f^\ast}{f}{\F}{\quant} \odot d_1(r,m) \odot d_2(m,f) \odot \homarg{r}{r^\ast}{\R}{\quant} \Big).
    \end{align*}
    By assuming all quantales considered in this paper are commutative (i.e., $\odot$ is symmetric), we can cleanly pair the terms and bound them using \Cref{eq:pf_d1_simp} and \Cref{eq:pf_d2_simp} via 
    \begin{align*}
         & \bigsqcup_{m \in \M} \Big( \bigl( d_1(r,m) \odot \homarg{r}{r^\ast}{\R}{\quant} \bigr) \odot \bigl( \homarg{f^\ast}{f}{\F}{\quant} \odot d_2(m,f) \bigr) \Big) \\
         & \sqsubseteq \bigsqcup_{m \in \M} \Big( d_1(r^\ast, m) \odot d_2(m, f^\ast) \Big)                                                                               \\
         & = (d_1;d_2)(r^\ast, f^\ast).
    \end{align*}
\end{proof}

\begin{lemma} \label{lem:q-parallel}
    The parallel composition defined in \eqref{eq:q-parallel} produces a valid $\quant$-design problem.
\end{lemma}
\begin{proof}
    Let $d: \R \xslashedrightarrow{} \F$ and $d': \R' \xslashedrightarrow{} \F'$ be two $\quant$-design problems. Following \Cref{def:dp}, they satisfy
    \begin{align}
        \homarg{f^\ast}{f}{\F}{\quant} \odot d(r,f) \odot \homarg{r}{r^\ast}{\R}{\quant}          & \sqsubseteq d(r^\ast,f^\ast) \label{eq:pf_parallel_d}         \\
        \homarg{f'^\ast}{f'}{\F'}{\quant} \odot d'(r',f') \odot \homarg{r'}{r'^\ast}{\R'}{\quant} & \sqsubseteq d'(r'^\ast,f'^\ast) \label{eq:pf_parallel_dprime}
    \end{align}
    for all valid objects in their respective categories.

    We must show that the parallel composition $d_{\otimes} \coloneqq (d \otimes d')$ satisfies the profunctor monotonicity condition on the product categories $\R \otimes \R'$ and $\F \otimes \F'$, namely that 
    \begin{multline*}
        \homarg{(f^\ast,f'^\ast)}{(f,f')}{\F \otimes \F'}{\quant} \odot d_{\otimes}\bigl((r,r'),(f,f')\bigr) \\
        \odot \homarg{(r,r')}{(r^\ast,r'^\ast)}{\R \otimes \R'}{\quant} \sqsubseteq d_{\otimes}\bigl((r^\ast,r'^\ast),(f^\ast,f'^\ast)\bigr).
    \end{multline*}

    Expanding this expression via the definition of the tensor product of $\quant$-categories (\Cref{def:tensor-product}) and parallel composition (\Cref{eq:q-parallel}), we have
    \begin{multline*}
        \Big( \homarg{f^\ast}{f}{\F}{\quant} \odot \homarg{f'^\ast}{f'}{\F'}{\quant} \Big) \odot \Big( d(r,f) \odot d'(r',f') \Big) \\
        \odot \Big( \homarg{r}{r^\ast}{\R}{\quant} \odot \homarg{r'}{r'^\ast}{\R'}{\quant} \Big).
    \end{multline*}

    Using the associativity and commutativity of the quantale multiplication ($\odot$), we can rearrange the terms into pairs corresponding to the base components $d$ and $d'$ as
    \begin{align*}
        = {} & \Big( \homarg{f^\ast}{f}{\F}{\quant} \odot d(r,f) \odot \homarg{r}{r^\ast}{\R}{\quant} \Big)                 \\
             & \odot \Big( \homarg{f'^\ast}{f'}{\F'}{\quant} \odot d'(r',f') \odot \homarg{r'}{r'^\ast}{\R'}{\quant} \Big).
    \end{align*}

    By applying the monotonicity constraints from \eqref{eq:pf_parallel_d} and \eqref{eq:pf_parallel_dprime}, we can bound this product as 
    \begin{align*}
         & \sqsubseteq d(r^\ast,f^\ast) \odot d'(r'^\ast,f'^\ast)               \\
         & = d_{\otimes}\bigl((r^\ast,r'^\ast),(f^\ast,f'^\ast)\bigr).
    \end{align*}
\end{proof}
\begin{lemma} \label{lem:q-feedback}
    The feedback operation defined in \eqref{eq:q-feedback} produces a valid $\quant$-design problem.
\end{lemma}
\begin{proof}
    Let $d: \R \otimes \M \xslashedrightarrow{} \F \otimes \M$ be a $\quant$-design problem. By Definition~\ref{def:dp}, the profunctor monotonicity condition on product spaces dictates that for all $r, r^\ast \in \R$, $m_1, m_2, m_1^\ast, m_2^\ast \in \M$, and $f, f^\ast \in \F$ we have 
    \begin{align*}
         & \homarg{(f^\ast,m_2^\ast)}{(f,m_2)}{\F \otimes \M}{\quant} \odot d\bigl((r,m_1),(f,m_2)\bigr) \\
         & \qquad \odot \homarg{(r,m_1)}{(r^\ast,m_1^\ast)}{\R \otimes \M}{\quant}                       \\
         & \sqsubseteq d\bigl((r^\ast,m_1^\ast),(f^\ast,m_2^\ast)\bigr).
    \end{align*}
    Expanding the tensor product hom-objects (\Cref{def:tensor-product}) yields
    \begin{align*}
         & \Big(\homarg{f^\ast}{f}{\F}{\quant} \odot \homarg{m_2^\ast}{m_2}{\M}{\quant}\Big) \odot d\bigl((r,m_1),(f,m_2)\bigr) \\
         & \qquad \odot \Big(\homarg{r}{r^\ast}{\R}{\quant} \odot \homarg{m_1}{m_1^\ast}{\M}{\quant}\Big)                       \\
         & \sqsubseteq d\bigl((r^\ast,m_1^\ast),(f^\ast,m_2^\ast)\bigr).
    \end{align*}
    To analyze the components forming $\Tr_{\M}(d)$, we configure the free variables by assigning $m_1^\ast = m_1$ and $m_2^\ast = m_2$. Using the identity axioms $e \sqsubseteq \homarg{m_1}{m_1}{\M}{\quant}$ and $e \sqsubseteq \homarg{m_2}{m_2}{\M}{\quant}$, we drop these terms to tightly bound the expression via 
    \begin{multline} \label{eq:pf_trace_simp}
        \homarg{f^\ast}{f}{\F}{\quant} \odot d\bigl((r,m_1),(f,m_2)\bigr) \odot \homarg{r}{r^\ast}{\R}{\quant} \\ \sqsubseteq d\bigl((r^\ast,m_1),(f^\ast,m_2)\bigr).
    \end{multline}

    We must prove that $\Tr_{\M}(d)$ satisfies profunctor monotonicity, namely
    \begin{multline} \label{eq:pf_trace_goal}
        \homarg{f^\ast}{f}{\F}{\quant} \odot \Tr_{\M}(d)(r,f) \odot \homarg{r}{r^\ast}{\R}{\quant} \\ \sqsubseteq \Tr_{\M}(d)(r^\ast,f^\ast).
    \end{multline}
    Substituting the trace formula in \eqref{eq:q-feedback} into the left-hand side of \eqref{eq:pf_trace_goal} and exploiting distributivity over $\bigsqcup$, we have 
    \begin{align*}
         & \homarg{f^\ast}{f}{\F}{\quant} \odot \Bigg( \bigsqcup_{m_1,m_2 \in \M} d\bigl((r,m_1),(f,m_2)\bigr)  \\
         & \qquad \qquad \odot \homarg{m_1}{m_2}{\M}{\quant} \Bigg) \odot \homarg{r}{r^\ast}{\R}{\quant}        \\
         & = \bigsqcup_{m_1,m_2 \in \M} \Big( \homarg{f^\ast}{f}{\F}{\quant} \odot d\bigl((r,m_1),(f,m_2)\bigr) \\
         & \qquad \qquad \odot \homarg{m_1}{m_2}{\M}{\quant} \odot \homarg{r}{r^\ast}{\R}{\quant} \Big).
    \end{align*}
    Using commutativity of $\odot$, we regroup the elements into the bounded unit from \eqref{eq:pf_trace_simp} to reach 
    \begin{align*}
         & = \bigsqcup_{m_1,m_2 \in \M} \Big( \homarg{f^\ast}{f}{\F}{\quant} \odot d\bigl((r,m_1),(f,m_2)\bigr) \odot \homarg{r}{r^\ast}{\R}{\quant} \Big) \\
         & \qquad \qquad \odot \homarg{m_1}{m_2}{\M}{\quant}                                                                                               \\
         & \sqsubseteq \bigsqcup_{m_1,m_2 \in \M} d\bigl((r^\ast,m_1),(f^\ast,m_2)\bigr) \odot \homarg{m_1}{m_2}{\M}{\quant}                               \\
         & = \Tr_{\M}(d)(r^\ast, f^\ast). 
    \end{align*}
\end{proof}

\subsection*{Proof of \Cref{lem:change-of-base-cat}}

For one direction, suppose $\phi$ is lax and $\cat{C}$ is a $\quant$-category.
For the identity axiom, combining the lax unit condition $e' \sqsubseteq' \phi(e)$ with the $\quant$-category identity $e \sqsubseteq \homarg{x}{x}{\cat{C}}{\quant}$ and monotonicity of $\phi$, we obtain
\begin{equation*}
    e' \sqsubseteq' \phi(e) \sqsubseteq' \phi\bigl(\homarg{x}{x}{\cat{C}}{\quant}\bigr) = \homarg{x}{x}{\phi_\ast\cat{C}}{\quant'}.
\end{equation*}
For the composition axiom, we apply lax multiplicativity and then monotonicity to the $\quant$-category composition law to find 
\begin{align*}
     & \homarg{x}{y}{\phi_\ast\cat{C}}{\quant'} \odot' \homarg{y}{z}{\phi_\ast\cat{C}}{\quant'}               \\
     & = \phi\bigl(\homarg{x}{y}{\cat{C}}{\quant}\bigr) \odot' \phi\bigl(\homarg{y}{z}{\cat{C}}{\quant}\bigr) \\
     & \sqsubseteq' \phi\bigl(\homarg{x}{y}{\cat{C}}{\quant} \odot \homarg{y}{z}{\cat{C}}{\quant}\bigr)
    \sqsubseteq' \phi\bigl(\homarg{x}{z}{\cat{C}}{\quant}\bigr)                                               \\
     & = \homarg{x}{z}{\phi_\ast\cat{C}}{\quant'}.
\end{align*}

For the other direction, suppose $\phi_{\ast}\cat{C}$ is a $\quant'$-category for every $\quant$-category $\cat{C}$.
For the unit condition, consider the single-object $\quant$-category $\cat{E}$ with $\homarg{\star}{\star}{\cat{E}}{\quant} \coloneqq e$. This is a valid enriched category since $e \sqsubseteq e$ and $e \odot e = e \sqsubseteq e$. Applying the identity axiom of $\phi_{\ast}\cat{E}$ yields $e' \sqsubseteq' \phi(e)$.

For the multiplicativity condition, fix arbitrary $q_1, q_2 \in Q$ and consider the three-object $\quant$-category $\cat{W}$ on the set $\{a, b, c\}$ with
\begin{equation*}
    \homarg{a}{b}{\cat{W}}{\quant} = q_1, \quad
    \homarg{b}{c}{\cat{W}}{\quant} = q_2, \quad
    \homarg{a}{c}{\cat{W}}{\quant} = q_1 \odot q_2,
\end{equation*}
with diagonal elements $\homarg{x}{x}{\cat{W}}{\quant} = e$, and all remaining elements given by $\bot = \bigsqcup \emptyset$. By virtue of the fact that $\bot$ is absorbing for $\odot$ (since $q \odot \bot = q \odot \bigsqcup \emptyset = \bigsqcup \emptyset = \bot$), every composition axiom involving a $\bot$-valued hom reduces to $\bot \sqsubseteq q$, which holds universally. Applying the composition axiom of $\phi_{\ast}\cat{W}$ to the triple $(a, b, c)$ yields
\begin{equation*}
    \phi(q_1) \odot' \phi(q_2) \sqsubseteq' \phi(q_1 \odot q_2). 
\end{equation*}

\subsection*{Proof of \cref{lem:heterogenous-co-design}}

By \cref{lem:change-of-base-cat}, $(\phi_{f})_\ast\F$ and $(\phi_{r})_\ast\R$ are both $\quant$-categories, with hom-values $\homarg{f^\ast}{f}{(\phi_{f})_\ast\F}{\quant} = \phi_f(\homarg{f^\ast}{f}{\F}{\quant_f})$ and $\homarg{r}{r^\ast}{(\phi_{r})_\ast\R}{\quant} = \phi_r(\homarg{r}{r^\ast}{\R}{\quant_r})$.
According to \eqref{eq:bimodule}, to make $d$ a valid $\quant$-design problem we need
\begin{equation*}
    \homarg{f^\ast}{f}{(\phi_{f})_\ast\F}{\quant} \odot d(r,f) \odot \homarg{r}{r^\ast}{(\phi_{r})_\ast\R}{\quant} \sqsubseteq d(r^\ast,f^\ast),
\end{equation*}
which is exactly \eqref{eq:heterogenous-co-design}. 

\subsection*{Proof of \Cref{thm:hetero-composition}}

We prove that each heterogeneous composition operation produces a valid profunctor.

\begin{lemma}
    The heterogeneous series composition defined in \eqref{eq:hetero-series} produces a valid $\quant_{\mathsf{eval}}$-design problem.
\end{lemma}
\begin{proof}
    Let $d: \R \xslashedrightarrow{}_{\quant_d} \M$ and $e: \M \xslashedrightarrow{}_{\quant_e} \F$ be heterogeneous design problems. We must show that $(d; e): (\phi_d)_\ast\R \xslashedrightarrow{}_{\quant_{\mathsf{eval}}} (\phi_e)_\ast\F$ satisfies the $\quant_{\mathsf{eval}}$-profunctor condition
    \begin{multline*}
        \homarg{f^\ast}{f}{(\phi_e)_\ast\F}{\quant_{\mathsf{eval}}} \odot_{\mathsf{eval}} (d;e)(r,f) \odot_{\mathsf{eval}} \homarg{r}{r^\ast}{(\phi_d)_\ast\R}{\quant_{\mathsf{eval}}} \\ \sqsubseteq_{\mathsf{eval}} (d;e)(r^\ast,f^\ast).
    \end{multline*}
    Substituting the definitions of pushforward categories and series composition, the left-hand side is
    \begin{align*}
         & \phi_e(\homarg{f^\ast}{f}{\F}{\quant_e}) \odot_{\mathsf{eval}} \left( \bigsqcup_{m \in \M} \phi_d(d(r, m)) \odot_{\mathsf{eval}} \phi_e(e(m, f)) \right) \\
         & \qquad \odot_{\mathsf{eval}} \phi_d(\homarg{r}{r^\ast}{\R}{\quant_d})                                                                                    \\
         & = \bigsqcup_{m \in \M} \left( \phi_d(d(r, m)) \odot_{\mathsf{eval}} \phi_d(\homarg{r}{r^\ast}{\R}{\quant_d}) \right)                                     \\
         & \qquad \odot_{\mathsf{eval}} \left( \phi_e(\homarg{f^\ast}{f}{\F}{\quant_e}) \odot_{\mathsf{eval}} \phi_e(e(m, f)) \right).
    \end{align*}
    By the laxity of $\phi_d$ and $\phi_e$ (Definition \ref{def:lax}), we have $\phi(q_1) \odot \phi(q_2) \sqsubseteq \phi(q_1 \odot q_2)$. Hence, 
    \begin{align*}
         & \sqsubseteq \bigsqcup_{m \in \M} \phi_d(d(r, m) \odot_d \homarg{r}{r^\ast}{\R}{\quant_d}) \odot_{\mathsf{eval}} \phi_e(\homarg{f^\ast}{f}{\F}{\quant_e} \odot_e e(m, f)).
    \end{align*}
    By the homogeneous profunctor conditions for $d$ and $e$, and the monotonicity of $\phi_d$ and $\phi_e$, we obtain
    \begin{align*}
         & \sqsubseteq \bigsqcup_{m \in \M} \phi_d(d(r^\ast, m)) \odot_{\mathsf{eval}} \phi_e(e(m, f^\ast)) = (d;e)(r^\ast, f^\ast).
    \end{align*}
\end{proof}

\begin{lemma}
    The heterogeneous parallel composition defined in \eqref{eq:hetero-parallel} produces a valid $\quant_{\mathsf{eval}}$-design problem.
\end{lemma}
\begin{proof}
    The proof follows identically by applying the laxity condition $\phi(q_1) \odot \phi(q_2) \sqsubseteq \phi(q_1 \odot q_2)$ to the product terms in \eqref{eq:hetero-parallel} and using the independent monotonicity of each constituent design problem.
\end{proof}

\begin{lemma}
    The heterogeneous feedback operation defined in \eqref{eq:hetero-trace} produces a valid $\quant_{\mathsf{eval}}$-design problem.
\end{lemma}
\begin{proof}
    By applying the trace formula in \eqref{eq:hetero-trace}, the profunctor condition is verified by distributing the scalarizing lax function $\phi_d$ over the internal composition of the loop. Specifically, the laxity of $\phi_d$ ensures that evaluating the feedback loop in $\quant_{\mathsf{eval}}$ is bounded by the pushforward of the loop's evaluation in $\quant_d$.
\end{proof}
%%%%%%%%%%%%%%%%%%%%%%%%%%%%%%

\end{document}